\documentclass[12pt]{iopart}
\usepackage[utf8]{inputenc}

\usepackage{amsthm}
\usepackage{amssymb}
\usepackage{amsbsy}
\usepackage{bbm}

\usepackage{enumerate}
\usepackage{graphicx}

\newtheorem{theorem}{Theorem}[section]
\newtheorem{corollary}[theorem]{Corollary}
\newtheorem{proposition}[theorem]{Proposition}

\theoremstyle{remark}
\newtheorem{remark}[theorem]{Remark}

\newcommand{\R}{\mathbb{R}}                 
\newcommand{\N}{\mathbb{N}}                 
\newcommand{\Z}{\mathbb{Z}}                 
\newcommand{\Q}{\mathbb{Q}}           

\newcommand{\dom}{D}                        
\newcommand{\D}{\mathcal{D}}                
\newcommand{\1}{\mathbbm{1}}               

\renewcommand{\mat}[1]{\boldsymbol{#1}}       
\newcommand{\norm}[1]{\left\|#1\right\|}        
\newcommand{\abs}[1]{\left|#1\right|}     
\newcommand{\scal}[1]{\left\langle#1\right\rangle}      

\newcommand{\text}[1]{\mathrm{#1}}
\newcommand{\mod}{\mathrm{\,mod\,}}           

\begin{document}
\title{Cantor spectra of magnetic chain graphs}
\author{Pavel Exner$^1$, Daniel Va\v{s}ata$^2$}
\address{$^1$ Doppler Institute for Mathematical Physics and Applied Mathematics, Czech
Technical University in Prague, B\v{r}ehov\'{a} 7, 11519 Prague, and Nuclear Physics
Institute ASCR, 25068 \v{R}e\v{z} near Prague, Czech Republic}
\address{$^2$ Department of Applied Mathematics, Faculty of Information Technology, Czech Technical University in Prague,
Th\'{a}kurova 9, Prague, 16000, Czech Republic}
\eads{\mailto{exner@ujf.cas.cz}, \mailto{daniel.vasata@fit.cvut.cz}}

\begin{abstract}
    We demonstrate a one-dimensional magnetic system can exhibit a
    Cantor-type spectrum using an example of a chain graph with
    $\delta$ coupling at the vertices exposed to a magnetic field
    perpendicular to the graph plane and varying along the chain. If the
    field grows linearly with an irrational slope, measured in terms of
    the flux through the loops of the chain, we demonstrate the
    character of the spectrum relating it to the almost Mathieu
    operator.
\end{abstract}


%
\vspace{2pc} \noindent{\it Keywords}: quantum chain graph, magnetic
field, almost Mathieu operator, Cantor spectrum

%
\submitto{\JPA}

%
%
%

%
\section{Introduction}\label{sec:introduction}

The observation that spectra of quantum system may exhibit fractal
properties was made first by Azbel \cite{Azbel64} but it really
caught the imagination when Hofstadter \cite{Hofstadter1976} made
the structure visible; then it triggered a long and fruitful
investigation of this phenomenon. On the mathematical side the
question was translated into the analysis of the almost Mathieu
equation which culminated recently in the proof of the ``Ten Martini
Problem'' by Avila and Jitomirskaya \cite{Avila2009}. On the
physical side, the effect remained theoretical for a long time.
Since the mentioned seminal papers, following an earlier work of
Peierls \cite{Peierls33} and Harper \cite{Harper55}, the natural
setting considered was a lattice in a homogeneous magnetic field
because it provided the needed two length scales, generically
incommensurable, from the lattice spacing and the cyclotron radius.
The spectrum of the corresponding Hamiltonian has fractal properties
as one can establish rigorously \cite{Bruning2007} by adapting deep
results about appropriate difference operators
\cite{Avila06,Avila2009}. It was not easy to observe the effect,
however, and the first experimental demonstration of such a spectral
character was done instead in a microwave waveguide system with
suitably placed obstacles simulating the almost Mathieu relation
\cite{Kuhl98}. Only recently an experimental realization of the
original concept was achieved using a graphene lattice
\cite{Dean2013,Ponomarenko2013}.

The aim of this note is to show that fractal spectra can arise also
in magnetic systems extended in a single direction only under two
conditions: the structure should have a nontrivial topology and the
magnetic field should vary along it. We are going to demonstrate
this claim using a simple example of a chain graph consisting of an
array of identical rings connected at the vertices in the simplest
nontrivial way known as the $\delta$ coupling and exposed to the
magnetic field perpendicular to the graph plane the intensity of
which increases linearly along the chain, with the slope $\alpha$
measured in terms of the number of the flux quanta through the ring.
This is the decisive quantity. It turns out that when $\alpha$ is
rational, the spectrum has a band-gap structure which allows for
description in terms of the Floquet-Bloch theory. On the other hand,
when $\alpha$ is irrational, the spectrum is a Cantor set, that is,
a nowhere dense closed set without isolated points. The way to prove
these results is to translate the original spectral problem into an
equivalent one involving a suitable self-adjoint operator on
$\ell^2(\Z)$ which is a useful and well-known trick in the quantum
graph theory, see e.g.
\cite{Cattaneo1997,Exner1997,Pankrashkin2012}. As a result, in the
rational case we rephrase the question as spectral analysis of a
simple Laurent operator, while in the irrational case we reduce the
problem to investigation of the almost Mathieu operator, for which
the Cantor property of the spectrum is known as mentioned above
\cite{Avila2009}.

Let us briefly describe the contents of the paper. In the next
section we will define properly the operator that serves as the
magnetic chain Hamiltonian. In Sec.~\ref{sec:duality_disc_opp} we
explain our main technical tool, a duality between the quantum graph
in question and an appropriate Jacobi operator, whose spectrum is
described in Sec.~\ref{sec:spectrum_general}. Finally,
Sec.~\ref{sec:spectrum_linear} contains our main result with some
corollaries and a discussion; it is followed by a few concluding
remarks.

%
\section{Magnetic chain graph}\label{sec:mgchain}

Quantum graphs, which is a short name for Schr\"odinger operators
the configuration space of which has the structure of a metric
graph, are an important class of models in quantum physics. They are
interesting both physically as models of various nanostructures, as
well as from the viewpoint of their mathematical properties; we
refer the reader to the recent monograph of Berkolaiko and Kuchment
\cite{Berkolaiko2013} for a thorough presentation and a rich
bibliography. One important class is represented by magnetic quantum
graphs, cf. for instance \cite{Kostrykin2003}.

Let us describe the particular system we will be interested in. It
is a metric graph $\Gamma$ consisting of an infinite linear chain of rings of
unit radius, centred at equally spaced points laying at a straight
line and touching their neighbours at the left and right. The
vertices are parametrized by integers $\Z$ and both the upper edge
$e_j^U$ and lower edge $e_j^L$ connecting the $j$-th vertex $v_j$
and $(j+1)$-th vertex $v_{j+1}$, which forms the $j$-th ring of the
graph, are parametrized by intervals $(0, \pi)$ directed along the
chain. Thus, if the initial vertex of an edge $e$ is denoted by
$\iota e$ and the terminal vertex by $\tau e$, then $\iota e_j^U =
\iota e_j^L = v_j$ and $\tau e_j^U = \tau e_j^L = v_{j+1}$. We
assume that the system is exposed to a magnetic field perpendicular
to the graph plane, which in contrast to \cite{Exner2015} is not
homogenerous but may vary along the chain. The Hamiltonian is the
graph version of the magnetic Schr\"odinger operator acting as
$\frac{1}{2m}\big(-i\hbar\nabla- \frac{e}{c}A \big)^2$ at each edge,
where $A$ stands for the tangential component of the corresponding
vector potential at a given point. However, it is known that in a
magnetic chain there are only the fluxes through the loops that
count, see \cite[Corollary 2.6.3]{Berkolaiko2013}, and therefore we
may, without loss of generality, choose a gauge in which the
(tangent component of the) vector potential $A$ is constant at each
particular ring; we denote by $A_j\in \R$ its value at the $j$-th
ring and by $\mat A = \{A_j\}_{j \in \Z}$  the sequence of all such
local vector potentials.

The state Hilbert space corresponding to a non-relativistic charged
spinless particle confined to the graph $\Gamma$ is $L^2(\Gamma)$.
For a function $\psi \in L^2(\Gamma)$ we further denote its
components on the upper and lower semicircles $e_j^U$ and $e_j^L$ of
the $j$-th ring by $\psi_j^U$ and $\psi_j^L$, respectively. The
whole system is depicted in Figure \ref{fig:chain_graph}.
\begin{figure}
    \begin{center}
        \includegraphics{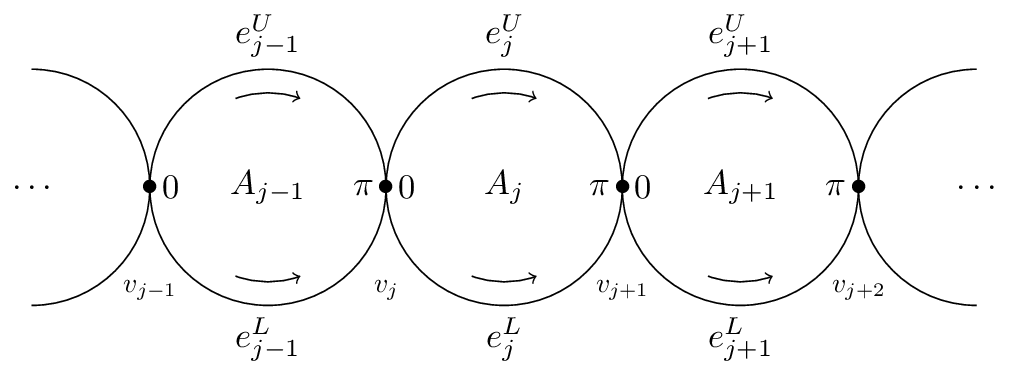}
    \end{center}
    \caption{Schematic depiction of the magnetic chain graph $\Gamma$.}
    \label{fig:chain_graph}
\end{figure}
Since the actual values of physical quantities will play no role in
the discussion we employ the rational system of units putting
$\hbar=2m=1$ and $\frac{e}{c}=1$. The Hamiltonian is then simply
$-\Delta_{\gamma, \mat A} = -\D^2$, where $\D$ is the
quasi-derivative which depends locally on the parametrisation of the
edge and the magnetic field; specifically, on the upper and lower
semicircles of the $j$-th chain ring, $\psi^U_j$ and $\psi^L_j$, it
acts as
\begin{equation*}
    \D \psi^U_j = (\psi^U_j)' + i A_j \psi^U_j\quad \text{and}\quad \D \psi^L_j = (\psi^L_j)' - i A_j \psi^L_j,
\end{equation*}
respectively.

In order to make $-\Delta_{\gamma, \mat A}$ a well-defined
self-adjoint operator we have to specify its domain which entails
choosing the boundary conditions satisfied by the functions at the
vertices of $\Gamma$, in physical terms this means to indicate the
coupling between the rings. We choose for the latter the simplest
nontrivial coupling commonly known as $\delta$. The domain
$\dom(-\Delta_{\gamma, \mat A})$ then consists of all functions from
the Sobolev space $H^2(\Gamma)$ satisfying at the graph vertices the
conditions
\begin{equation}\label{eq:delta_cond_1}
    \psi^U_{j}(0_+) = \psi^L_{j}(0_+) = \psi^U_{j-1}(\pi_-) = \psi^L_{j-1}(\pi_-),
\end{equation}
\begin{equation}\label{eq:delta_cond_2}
    -\D\psi^U_{j-1}(\pi_-) - \D\psi^L_{j-1}(\pi_-) + \D \psi^U_{j}(0_+) + \D\psi^L_{j}(0_+)  = \gamma \psi^U_{j}(0_+)
\end{equation}
for all $j \in \Z$, where $\gamma$ is the coupling constant and
$\psi^U_{j}(0_+)$ is the right limit of $\psi^U_{j}(x)$ at zero and
$\psi^U_j(\pi_-)$ is the left limit of $\psi^U_{j-1}(x)$ at the
point $\pi$, etc. Note the different signs of the quasiderivative
$\D$ at $0_+$ and $\pi_-$ which reflects the fact that the one-sided
derivative at a vertex should be taken in the outgoing direction.

%
%
\section{Duality with a discrete operator}\label{sec:duality_disc_opp}

In order to obtain the spectrum of $-\Delta_{\gamma, \mat A}$ we
employ a particular kind of the duality mentioned in the
introduction, relating it to the difference operator $L_{\mat A}$
which is a bounded self-adjoint operator on $\ell^2(\Z)$ defined by
\begin{equation}\label{eq:def_discrete_op}
    (L_{\mat A} \varphi)_j = 2\cos(A_j \pi)\varphi_{j+1} + 2\cos(A_{j-1} \pi) \varphi_{j-1}.
\end{equation}
We employ the results obtained by K.~Pankrashkin in \cite[Section
2.3]{Pankrashkin2012} using the boundary triple technique, see also
\cite{Bruning2008}.

To begin with, consider the local gauge transform $G$ given by
\[
  \big(G\psi^U_{j}\big)(x) = e^{-i x A_j} \psi^U_{j}(x)\quad \text{and}\quad \big(G\psi^L_{j}\big)(x) = e^{i x A_j} \psi^L_{j}(x).
\]
Using it, the operator $-\Delta_{\gamma, \mat A}$ is unitarily
equivalent to the operator $H_{\gamma, \mat A}$ on $L^2(\Gamma)$
acting as
\[
  H_{\gamma, \mat A}\psi^U_{j} = - (\psi^U_j)'' \quad \text{and}\quad H_{\gamma, \mat A}\psi^L_{j} = - (\psi^L_j)''
\]
with the domain $\dom(H_{\gamma, \mat A})$ consisting of the
functions from $H^2(\Gamma)$ that obey the boundary conditions
\begin{equation}\label{eq:h_delta_cond_1}
    \psi^U_{j}(0_+) = \psi^L_{j}(0_+) = e^{-i \pi A_{j-1}} \psi^U_{j-1}(\pi_-) = e^{i \pi A_{j-1}}\psi^L_{j-1}(\pi_-),
\end{equation}
\begin{equation}\label{eq:h_delta_cond_2}
   \hspace{-5   em} - e^{-i \pi A_{j-1}} (\psi^{U}_{j-1})'(\pi_-) - e^{i \pi A_{j-1}} (\psi^L_{j-1})'(\pi_-) + (\psi^U_{j})'(0_+) + (\psi^L_{j})'(0_+)  = \gamma \psi^U_{j}(0_+).
\end{equation}
Consider next the solutions $s$ and $c$ of the differential equation
$-y'' - z y = 0$ satisfying the boundary conditions
\[
  s(0_+;z) = c'(0_+;z) = 0\quad \text{and}\quad s'(0_+;z) = c(0_+;z) = 1.
\]
They are given explicitly by
\[
  s(x;z) =
    \cases{\frac{\sin(x\sqrt{z})}{\sqrt{z}}& for $z \neq 0$,\\
        x& for $z = 0$,}
  \quad \text{and}\quad
  c(x;z) = \cos(x\sqrt{z}),
\]
where $\sqrt{z}$ stands for the principal branch of the square root.
For $z=k^2$ where $k\in \N$, $s(x;k^2)$ is actually a solution on
$(0,\pi)$ satisfying the Dirichlet boundary conditions. We denote
the set of all such $z$ by $\sigma_D$, i.e. $\sigma_D = \{k^2 |\, k
\in \N\}$. The following proposition shows that all those points are
actually eigenvalues of $H_{\gamma, \mat A}$ and thus also of
$-\Delta_{\gamma, \mat A}$.
\begin{proposition}\label{prop:spectrum_natural}
    Let $k \in \N$ and $j \in \Z$ such that $A_j \in \Z$ or $A_{j-1}, A_{j} \notin \Z$.
    Then there exists an eigenvector $\psi(k^2, j)$ of $H_{\gamma, \mat A}$ corresponding to the eigenvalue $k^2$
    which can be described as follows:
    \begin{enumerate}[a)]
        \item If $A_j \in \Z$ then
        \begin{eqnarray*}
            \psi_l^U(x; k^2,j) =
                \cases{s(x;k^2)& for $l = j$,\\
                    0& for $l \neq j$,}
            \\[1ex]
            \psi_l^{L}(x; k^2,j) =
                \cases{-s(x;k^2)& for $l = j$,\\
                    0& for $l \neq j$}
        \end{eqnarray*}
        for all $l \in \Z$.
        \item If $A_{j-1}, A_{j} \notin \Z$ then
        \begin{eqnarray*}
            \psi_l^U(x; k^2,j) =&
                \cases{\phantom{-}\sin(A_j \pi) \cdot s(x;k^2)& for $l = j - 1$,\\
                    -\sin(A_{j-1} \pi)\cdot e^{i\pi A_j} \cdot s(x-\pi;k^2)& for $l = j$,\\
                    \phantom{-}0& elsewhere,}
            \\[1ex]
            \psi_l^{L}(x; k^2,j) =&
                \cases{-\sin(A_j \pi) \cdot s(x;k^2)& for $l = j - 1$,\\
                    \phantom{-}\sin(A_{j-1} \pi) \cdot e^{-i\pi A_j} \cdot s(x-\pi;k^2)& for $l = j$,\\
                    \phantom{-}0& elsewhere}
        \end{eqnarray*}
        for all $l \in \Z$.
    \end{enumerate}
\end{proposition}
\begin{proof}
    In both cases the functions $\psi(k^2, j)$ specified above clearly satisfy boundary conditions
    \eref{eq:h_delta_cond_1} and \eref{eq:h_delta_cond_2}.
\end{proof}

Before we state the main spectral equivalence, we have to introduce
another piece of notation. Let $H$ be a self adjoint operator and $B
\subset \R$ a Borel set. By $H_B$ we denote the part of $H$ in $B$
referring to the spectral projection $\1_{B}(H)$ of $H$ to the set
$B$, in other words, $H_B = H \1_{B}(H)$.
\begin{theorem}\label{th:dual_unitary_rel}
  For any interval $J \subset \R\setminus \sigma_D$, the operator $(H_{\gamma, \mat A})_J$ is unitarily equivalent
  to the preimage
  $\eta^{(-1)}\big((L_{\mat A})_{\eta(J)}\big)$, where
  \[
    \eta(z) = \gamma s(\pi; z) + 2 c(\pi; z) + 2s'(\pi; z).
  \]
\end{theorem}
\begin{proof}

The claim follows from Theorem 18 in \cite{Pankrashkin2012} where
both $\eta$ and the difference operator $L_{\mat A}$ are divided by
four.
\end{proof}

\noindent Using the expressions for $s(x;z)$ and $c(x;z)$ we can write
\begin{equation}\label{eq:def_eta}
  \eta(z) =
    \cases{\gamma\, \frac{\sin\big(\pi\sqrt{z}\big)}{\sqrt{z}} + 4 \cos\big(\pi\sqrt{z}\big)& for $z \neq 0$,\\
        \gamma \pi + 4& for $z = 0$.}
\end{equation}
Thus, owing to the fact that $H_{\gamma, \mat A}$ is unitarily
equivalent to $-\Delta_{\gamma, \mat A}$, the spectrum of
$-\Delta_{\gamma, \mat A}$ is related to the spectrum of $L_{\mat
A}$ via the preimage of $\sigma(L_{\mat A})$ under the entire
function $\eta$. This means that, up to the discrete set $\{n^2|\,n
\in \N\}$ of infinitely degenerate eigenvalues of $-\Delta_{\gamma,
\mat A}$ which are described in Proposition
\ref{prop:spectrum_natural}, one has $\lambda \in \sigma(L_{\mat
A})$ if and only if $\eta^{(-1)}\big(\{\lambda\}\big) = \{z\in\R |\,
\eta(z) = \lambda\} \subset \sigma(-\Delta_{\gamma, \mat A})$.
Moreover, $\lambda \in \sigma(L_{\mat A})$ is an eigenvalue if and
only if points from $\eta^{(-1)}\big(\{\lambda\}\big)$ are
eigenvalues and the same holds for the other parts of the spectrum,
e.g. for the essential, absolutely continuous, and singular
continuous spectral component.
\begin{figure}
    \begin{center}
        \includegraphics{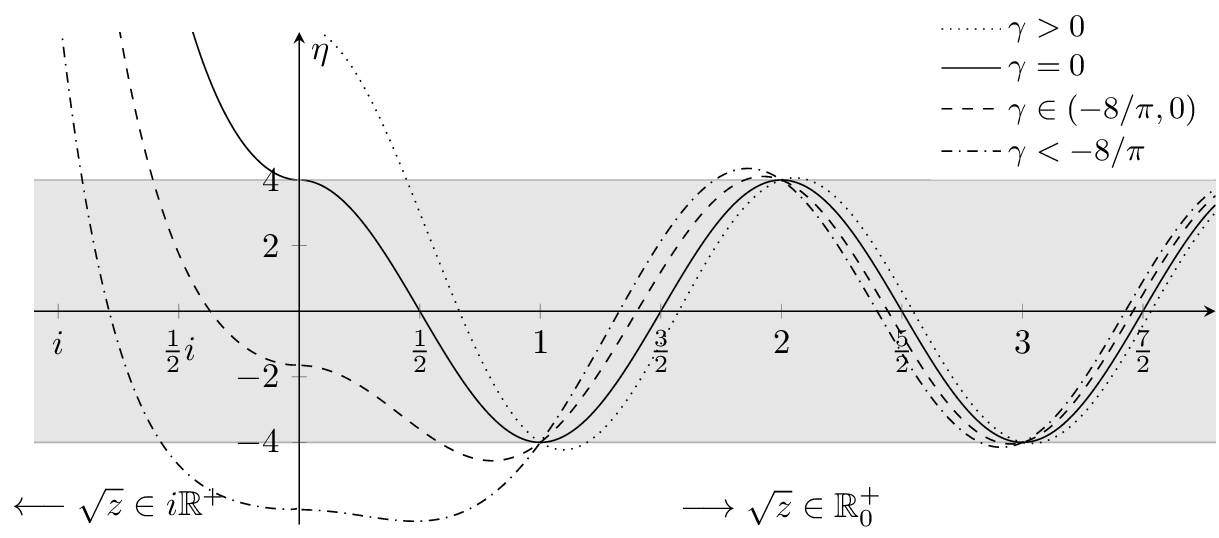}
    \end{center}
    \caption{The influence of the parameter $\gamma$ on the behaviour of
    $\eta(z)$ which is plotted as a function of $\sqrt{z}$. On the right side of the vertical axis we plot the positive increasing
    values of $\sqrt{z}$ and on the left side we plot increasing values of the purely imaginary positive values of $\sqrt{z}$,
    i.e. of $\sqrt{z} = i\kappa$, $\kappa > 0$.}
    \label{fig:eta_fun}
\end{figure}

Clearly, $\norm{L_{\mat A}} \leq 4$, where $\norm{\cdot}$ is the
operator norm on $\ell^2(\Z)$, and therefore $\sigma(L_{\mat A})
\subset [-4,4]$. We are thus interested in the behaviour of $\eta$
when its values are inside the interval $[-4,4]$. This is shown in
Figure \ref{fig:eta_fun}. The function $\eta$ is continuous with the
continuous derivative bounded in each interval $(x,\infty), x\in
\R$, and it behaves essentially in the same way in each of the
intervals $\big[n^2, (n+1)^2\big], n \in \N$. Let $I_n :=
\eta^{(-1)}\big([-4,4]\big) \cap \big(n^2, (n+1)^2\big)$ be the
preimage of $[-4,4]$ restricted to $\big(n^2,(n+1)^2\big)$. By
inspecting the derivative of $\eta$, it is easy to check that $I_n$
is always an interval. Moreover, for $\gamma > 0$ we have $I_n  =
\big[a_n, (n+1)^2\big)$ where $n^2 < a_n < (n+1)^2$. On the other
hand, for $\gamma < 0$ we have $I_n = (n^2, b_n]$ where $n^2 < b_n <
(n+1)^2$. Finally, $I_n = \big(n^2, (n+1)^2\big)$ holds for $\gamma
= 0$. Thus whenever $\gamma \neq 0$, the intervals $I_n$ are
separated by a positive distance, i.e. there are gaps between parts
of the spectrum.

For $z < 1$, the behaviour of $\eta(z)$ is slightly different and
much stronger influenced by the value of $\gamma$. If $\gamma > 0$
then $\eta$ is decreasing on $(-\infty, 1)$. If $-12/\pi \leq \gamma
\leq 0$ then it is decreasing up to a certain point in $(0,1)$ and
then increasing. Finally, if $\gamma < -12/\pi$ then $\eta$ is
decreasing up to some point in $(-\infty,0)$ and then increasing.
Let $I_0 := \eta^{(-1)}\big([-4,4]\big) \cap \big(-\infty, 1\big)$.
It clearly follows that $I_0$ is again an interval. Since $\lim_{z
\to -\infty} \eta(z) = \infty$, we obtain that for $\gamma > 0$,
$I_0  = [a_0, 1)$ where $0< a_0 <1$. For $\gamma = 0$ we have $I_0 =
[0,1)$. For $-8/\pi < \gamma < 0$, $I_0 = [a_0, b_0]$ where $a_0 < 0
< b_0 < 1$. For $\gamma = -8/\pi$, $I_0 = [a_0, 0]$ where $a_0 < 0$,
and finally, for $\gamma < -8/\pi$, $I_0 = [a_0, b_0]$ where $a_0 <
b_0 < 0$. Note that $0 \in I_0$ holds only when $\gamma \in [-8/\pi,
0]$. These findings combined with Proposition
\ref{prop:spectrum_natural} and Theorem \ref{th:dual_unitary_rel}
yield the following statement about the basic structure of the
spectrum of $-\Delta_{\gamma, \mat A}$.

\begin{proposition}\label{prop:spec:basic_for_original_op}
The spectrum of $-\Delta_{\gamma, \mat A}$ is bounded from below and
can be decomposed into the discrete set
    $\sigma_{D} = \{n^2|\,n \in \N\}$ of infinitely degenerate eigenvalues and the part $\sigma_{L_{\mat A}}$
    determined by the spectrum of $L_{\mat A}$,
    $\sigma(-\Delta_{\gamma, \mat A}) = \sigma_{p} \cup \sigma_{L_{\mat A}}$, where $\sigma_{L_{\mat A}}$ can be written as the union
    \[
        \sigma_{L_{\mat A}} = \bigcup_{n=0}^{\infty} \sigma_n
    \]
    with $\sigma_n = \eta^{(-1)}\big(\sigma(L_{\mat A})\big) \cap I_n$ for $n \geq 0$, $I_n = \eta^{(-1)}\big([-4,4]\big) \cap \big(n^2, (n+1)^2\big)$
    for $n > 0$, and $I_0 = \eta^{(-1)}\big([-4,4]\big) \cap \big(-\infty, 1\big)$.

    When $\gamma \neq 0$, the spectrum has always gaps between the $\sigma_n$'s. For $\gamma > 0$, the spectrum is positive.
    For $\gamma < -8\pi$, the spectrum has a negative part and does not contain zero. Finally, $0 \in \sigma(-\Delta_{\gamma, \mat A})$ if and only if
    $\gamma \pi + 4 \in \sigma(L_{\mat A})$.
\end{proposition}

%
%
\section{Spectrum in the general case}\label{sec:spectrum_general}

The main conclusion from the previous discussion is that in order to get a
full picture of the spectrum of $-\Delta_{\gamma, \mat A}$ we need
to investigate the spectrum of the bounded self-adjoint Jacobi
operator $L_{\mat A}$. Spectral analysis of Jacobi
operators is a well understood topic, see e.g. \cite{Teschl2000},
and we can pick the tools suitable for our present case.

Denoting $a_j := 2 \cos(A_j \pi)$ we can
express the action of $L_{\mat A}$ as
\[
    (L_{\mat A} \varphi)_j = a_j\varphi_{j+1} + a_{j-1}\varphi_{j-1}
\]
for any $\varphi \in \ell^2(\Z)$.
First thing to mention is that the spectrum of $L_{\mat A}$ does not
depend on the signs of $a_j$. This follows from the fact that
$L_{\mat A}$ is unitarily equivalent to $L_{\tilde{\mat A}}$
whenever $\abs{a_j} = \abs{\tilde a_j}$. It can be easily checked
that the equivalence is mediated by the unitary operator $U_{\mat A,
\tilde{\mat A}}$, i.e. $L_{\tilde{\mat A}} = U_{\mat A, \tilde{\mat
A}} L_{\mat A} U_{\mat A, \tilde{\mat A}}^{-1}$, defined by
\[
    \big(U_{\mat A, \tilde{\mat A}}\varphi \big)_j = u_j \varphi_j,
\]
for any $\varphi \in \ell^2(\Z)$, where
\[
    u_j =
    \cases{1& for $j = 0$,\\
        s_j s_{j-1} \ldots s_2 s_{1}& for $j > 0$,\\
        s_j s_{j+1} \ldots s_{-2} s_{-1}& for $j < 0$,}
    \quad\text{and} \quad
    s_j =
    \cases{1 & for $\tilde a_j = a_j$,\\
        -1 & otherwise.}
\]
This unitary invariance can be used to find upper and lower
bounds of the spectrum. By simple manipulations we get
\[
    \scal{\varphi,L_{\mat A} \varphi} = - \sum_{j \in \Z} a_j \abs{\varphi_{j+1} - \varphi_j}^2
    + \sum_{j \in \Z} (a_{j-1} + a_j) \abs{\varphi_j}^2.
\]
Let $\mat A^+$ be such that $a_j^+ = \abs{a_j}$, then we have
\begin{eqnarray*}
    \scal{\varphi,L_{\mat A} \varphi}
    &= \scal{U_{\mat A, \tilde{\mat A}}\varphi,L_{\mat A^+} U_{\mat A, \tilde{\mat A}}\varphi}
    \leq \sum_{j \in \Z} \big(\abs{a_{j-1}} + \abs{a_j}\big) \abs{\big(U_{\mat A, \tilde{\mat A}}\varphi \big)_j}^2\\
    &\leq \sup_{j\in\Z} c_j \norm{\varphi},
\end{eqnarray*}
where
\[
    c_j = \abs{a_{j-1}} + \abs{a_j} = 2\big(\abs{\cos(A_{j-1} \pi)} + \abs{\cos(A_{j}\pi)}\big).
\]
Similarly, using $\mat A^-$ such that $a_j^- = -\abs{a_j}$, we get
\[
    \scal{\varphi,L_{\mat A} \varphi}
    = \scal{U_{\mat A, \tilde{\mat A}}\varphi,L_{\mat A^-} U_{\mat A, \tilde{\mat A}}\varphi}
    \geq -\sup_{j\in\Z} c_j \norm{\varphi},
\]
which implies for the spectrum
\begin{equation}\label{eq:spec_LA_bound}
    \sigma(L_{\mat A}) \subset [-\sup_{j\in\Z}c_j, \sup_{j\in\Z}c_j].
\end{equation}
\begin{remark}
    It follows from the previous bounds that if $\sup_{j\in\Z}c_j < 4$,
    which means that all the pairs $A_{j-1}, A_j$ are uniformly separated from pairs of integers,
    the gaps between the parts $\sigma_n$ of the spectrum of $-\Delta_{\gamma, \mat A}$
    from Proposition \ref{prop:spec:basic_for_original_op} are always open and contain exactly one eigenvalue each.
\end{remark}

Let us turn to the situation, when some $a_j$'s are equal to zero, which
happens if the sequence $\{A_j\}$ contains half-integers. First
we introduce some notation, putting
\begin{eqnarray*}
    J_0 &:= \{j \in \Z \mid A_j + 1/2 \in \Z\},\\
    J &:= J_0 \cup \big(\{-\infty\} \setminus \inf J_0\big) \cup \big(\{\infty\} \setminus \sup J_0\big),
\end{eqnarray*}
i.e. $J$ contains $\infty$ whenever $J_0$ is bounded from above and
$-\infty$ whenever $J_0$ is bounded from below. We say that $j, k
\in J$ are \emph{neighbouring} in $J$ if $j < k$ and there is no
$i\in J$ such that $j < i < k$. For any $j, k \in J$ neighbouring in
$J$ let $L_{j,k}$ be the restriction of $L_{\mat A}$ to
$\{j+1,\ldots, k\}$. Clearly, $L_{j,k}$ is an operator on
$\ell^2(\{j+1,\ldots, k\})$ given by
\[
    (L_{j, k} \varphi)_i =
    \cases{a_{j+1} \varphi_{j+2} & for $i = j + 1$,\\
        a_j\varphi_{j+1} + a_{j-1}\varphi_{j-1} & for $j + 1 < i < k$,\\
        a_{k-1}\varphi_k& for $i = k$,}
\]
where $a_i \neq 0$ for all $j < i < k$. This allows us to write the
decomposition
\[
    L_{\mat A} = \bigoplus_{j,k \in J,\atop\text{neighbouring\, in}\, J} L_{j,k}.
\]
When $a_j \neq 0$ for all $j \in \Z$, then $J_0 = \emptyset$, $J =
\{-\infty,\infty\}$ and hence $L_{\mat A} = L_{-\infty,\infty}$.

\begin{theorem}\label{th:spec:periodic_degenerate}
    Under the previous notation
    \[
      \sigma(L_{\mat A}) = \overline{\bigcup_{j,k \in J,\atop\text{neighbouring\, in}\, J} \sigma(L_{j,k})}
    \]
    and the essential spectrum of $L_{\mat A}$ is nonempty.
    If $j,k \in J_0$ then $L_{j,k}$ has a pure point spectrum containing $k - j$ different eigenvalues.
    If $j = -\infty$ or $k = \infty$ then the spectrum  of $L_{j,k}$ has multiplicity at most two,
    that of the singular spectrum being one, and a nonempty essential part.
\end{theorem}
\begin{proof}
    The nonemptiness of the essential spectrum follows from boundedness of $L_{\mat A}$.
    When $j,k \in J_0$ the operator $L_{j,k}$ corresponds to a symmetric tridiagonal matrix $(k - j) \times (k-j)$
    with nonzero upper and lower diagonals which implies that it has $k-j$ different eigenvalues.
    When $j = -\infty$ or $k = \infty$ then the assertion follows from Theorem 3.4, Lemma 3.6
    in \cite{Teschl2000}, and the boundedness of $L_{j,k}$.
\end{proof}
Note that the absolutely continuous spectrum of $L_{\mat A}$, which
can be present only when $J_0$ is bounded from at least one side,
can be further determined by the principle of subordinacy, see e.g.
\cite[Section 3.3]{Teschl2000}.

Other interesting situation is the periodic one when there exists $N
\in \N$ such that $A_j = A_{j+N}$ holds for all $j \in \Z$ or more
generally, in view of the invariance of the spectrum w.r.t. the
signs of $a_j$, when $\abs{a_j} = \abs{a_{j+N}}$ holds for all $j
\in \Z$. If $a_j = 0$, or equivalently $A_j + 1/2 \in \Z$ for some
$j$, then the previous theorem implies that the spectrum is
trivially given by a finite number of eigenvalues with infinite
multiplicities. Otherwise, when $a_j \neq 0$ for all $j\in \Z$ one
may apply Floquet-Bloch theory to show that the spectrum is purely
absolutely continuous with a band-and-gap structure. The following
assertion summarizes the result proven e.g. in \cite[Sections 7.1
and 7.2]{Teschl2000}.
\begin{theorem}\label{th:spec:periodic_nondegenerate}
    Let $a_j \neq 0$ for all $j\in \Z$ and $\abs{a_j} = \abs{a_{j+N}}$ for some $N \in \N$ and all $j \in \Z$, i.e.
  $A_j + 1/2 \notin \Z$ and $\abs{\cos(A_j \pi)} = \abs{\cos(A_{j+N}\pi)}$,
  where $N$ is the smallest number with such property. Then
    the spectrum of $L_{\mat A}$ is purely absolutely continuous and consists
    of $N$ closed intervals possibly touching at the endpoints.
\end{theorem}

%
%
\section{A linear field growth}\label{sec:spectrum_linear}

Suppose now that $A_j = \alpha j + \theta$ holds for some $\alpha,
\theta \in \R$ and every $j\in\Z$. We denote the corresponding
operator $L_{\mat A}$ by $L_{\alpha,\theta}$, i.e.
\[
  (L_{\alpha,\theta} \varphi)_j = 2\cos\big(\pi(\alpha j + \theta)\big)\varphi_{j+1} + 2\cos\big(\pi(\alpha j - \alpha + \theta)\big) \varphi_{j-1}
\]
for all $j \in \Z$. Properties of the spectrum of
$L_{\alpha,\theta}$ are strongly influenced by number theoretic
properties of $\alpha$ and $\theta$. If $\alpha$ is a rational
number, $\alpha = p/q$, where $p$ and $q$ are relatively prime, then
$L_{\alpha,\theta}$ is, according to the discussion in the previous
section, periodic with the period $N = q$. Two distinct situations
may occur depending on the value of $\theta$.
\begin{theorem}\label{th:spec_LA_for_rationals}
  Assume that
  $\alpha = p/q$, where $p$ and $q$ are relatively prime. Then:
  \begin{enumerate}[(a)]
    \item If $\alpha j + \theta + \frac{1}{2} \notin \Z$ for all $j = 0,\ldots, q-1$, then $L_{\alpha,\theta}$ has purely absolutely
      continuous spectrum that consists of $q$ closed intervals possibly touching at the endpoints. In particular, $\sigma(L_{\alpha,\theta}) = \big[-4 \abs{\cos(\pi \theta)},4 \abs{\cos(\pi \theta)}\big]$ holds if $q = 1$.
    \item If $\alpha j + \theta + \frac{1}{2} \in \Z$ for some $j = 0,\ldots, q-1$, then the spectrum of $L_{\alpha,\theta}$  is of pure point type consisting  of
      $q$ distinct eigenvalues of infinite degeneracy. In particular, $\sigma(L_{\alpha,\theta}) = \{0\}$ holds if $q = 1$.
  \end{enumerate}
\end{theorem}
\begin{proof}
  Part (a) follows directly from Theorem \ref{th:spec:periodic_nondegenerate}.
  For $q = 1$ corresponding to $\alpha \in \Z$ the spectrum may be calculated directly, see e.g. \cite[Section 1.3]{Teschl2000}.

  In case (b) we may without loss of generality assume $\theta + \frac{1}{2} \in \Z$. Thus,
  $a_j = 2 \cos(A_j \pi) = 0$ for $j \mod q = 0$ and $a_j \neq 0$ otherwise. Hence, with the notation from the previous section,
  $J_0 = J = q\Z$ and $L_{j q,(j+1)q}$ are the same for all $j \in \Z$. This together with
  Theorem \ref{th:spec:periodic_degenerate} yields the assertion. If $q = 1$ we have $\alpha \in \Z$ and
  from the assumption $\theta + \frac{1}{2} \in \Z$ it follows that $a_j = 0$ holds for all $j \in \Z$, and consequently, $L_{\alpha,\theta}$ is a null operator.
\end{proof}
\begin{remark}
  Note that (a) occurs, for example, whenever $\theta$ is irrational.
\end{remark}

On the other hand, if $\alpha \notin \Q$ the spectrum of
$L_{\alpha,\theta}$ is closely related to the spectrum of the almost
Mathieu operator $H_{\alpha, \lambda, \theta}$ in the critical
situation, $\lambda = 2$, which for any $\alpha, \lambda, \theta \in
\R$ acts as
\[
    \big(H_{\alpha, \lambda, \theta} \varphi\big)_j
    = \varphi_{j+1} + \varphi_{j-1} + \lambda \cos(2\pi\alpha j + \theta)\varphi_j
\]
for any $\varphi \in \ell^2(\Z)$ and all $j\in\Z$. Recall that the
almost Mathieu operator is one of the most studied discrete
one-dimensional Schr\"odinger operator during several recent
decades, see e.g. \cite{Last2005} for a review. The spectrum of
$H_{\alpha, 2, \theta}$ as a set when $\alpha$ is irrational has
many interesting properties. First of all, it does not depend on
$\theta$, see \cite{Avron1983,Simon82}. Next, it is a Cantor set,
i.e. the perfect nowhere dense set; this property is known as the
``Ten Martini Problem''. The name of the challenge was coined by
Simon \cite{Simon82}, its proof was completed by Avila and
Jitomirskaya in \cite{Avila2009}. Moreover, the Lebesgue measure of
the spectrum of $H_{\alpha, 2, \theta}$ is zero, which is known as
Aubry-Andr\'{e} conjecture on the measure of the spectrum of the
almost Mathieu operator, demonstrated finally by Avila and Krikorian
in \cite{Avila06}. The picture arising from this survey can be
described as follows.
\begin{theorem}\label{th:almost_mathieu_spec_prop}
  For any $\alpha \notin \Q$, the spectrum of $H_{\alpha, 2, \theta}$ does not depend on $\theta$
  and it is a Cantor set of Lebesgue measure zero.
\end{theorem}

In order to reveal the relation between $L_{\alpha,\theta}$ and
$H_{\alpha, 2, \theta}$ we employ ideas from \cite{Shubin1994}. We
start by introducing the abstract Rotation Algebra $A_\alpha$ which
is a $C^*$ algebra generated by two unitary elements $u, v$ with the
commutation relation
\[
  uv = e^{i 2\pi \alpha} v u,
\]
see also \cite{Effros1967,Pedersen1979,Choi1990,Rieffel1981} for
more details. We can consider the representation $\pi_{\theta}$
generated by operators $U = \pi_{\theta}(u)$ and $V =
\pi_{\theta}(v)$,
\[
  (U \varphi)_j := \varphi_{j+1}, \quad (V \varphi)_j := e^{i 2\pi \alpha j + \theta} \varphi_j.
\]
Then the almost Mathieu operator coincides with the image of the
element
\[
  h_{\alpha} = u + u^{-1} + v + v^{-1} \in A_{\alpha},
\]
in other words, $H_{\alpha, 2, \theta} = \pi_{\theta}(h_{\alpha})$.
On the other hand, one can consider the representation
$\pi'_{\theta}$ generated by operators
\[
  (U \varphi)_j =  e^{i \pi (\alpha j + \theta)} \varphi_{j+1}, \quad (V \varphi)_j = e^{i \pi (\alpha (j-1) + \theta)} \varphi_{j-1}.
\]
In this case we have $L_{\alpha,\theta} =
\pi'_{\theta}(h_{\alpha})$.

When $\alpha \notin \Q$, it can be checked that $A_\alpha$ is
simple, see e.g. \cite{Effros1967,Pedersen1979,Power78}. This
implies that all its representations are faithful and thus they
preserve the spectrum of $h_{\alpha}$, which is defined as a set of
those complex $\lambda$ such that $h_{\alpha} - \lambda I$ is not
invertible, see e.g. \cite{Pedersen1979}. As a result, spectra of
$L_{\alpha,\theta}$ and $H_{\alpha, 2, \theta}$ as sets coincide,
\begin{equation}\label{eq:spec_set_equality_LA_mathieu}
  \sigma(L_{\alpha,\theta}) = \sigma(H_{\alpha, 2, \theta}),
\end{equation}
and are independent of $\theta$. This in combination with Theorem
\ref{th:almost_mathieu_spec_prop} proves the following assertion.
\begin{theorem}\label{th:spec_LA_irrational_cantor}
  For any $\alpha \notin \Q$, the spectrum of $L_{\alpha,\theta}$ as a set does not depend on $\theta$ and it is a Cantor set of Lebesgue measure zero.
\end{theorem}
\begin{remark}
  Note that all the previous considerations are equally valid for any $A_j$ such that
  $\abs{\cos(A_j \pi)} = \abs{\cos\big(\pi(\alpha j + \theta)\big)}$
  as a result of the invariance of the spectrum with respect to the signs of $a_j = \cos(A_j \pi)$ discussed in the previous section.
\end{remark}
As for the original operator $-\Delta_{\gamma, \mat A}$, we may
combine the previous observations to obtain the following theorem.
\begin{theorem}\label{th:spec_final_lin}
  Let $A_j = \alpha j + \theta$ for some $\alpha, \theta \in \R$ and every $j\in\Z$. Then for the spectrum $\sigma(-\Delta_{\gamma, \mat A})$
  the following holds:
  \begin{enumerate}[(a)]
    \item If $\alpha, \theta \in \Z$ and $\gamma = 0$, then
      $\sigma(-\Delta_{\gamma, \mat A}) = \sigma_{ac}(-\Delta_{\gamma, \mat A}) \cup \sigma_{pp}(-\Delta_{\gamma, \mat A})$ where
      $\sigma_{ac}(-\Delta_{\gamma, \mat A}) = [0,\infty)$ and $\sigma_{pp}(-\Delta_{\gamma, \mat A}) = \{n^2 | n\in\N\}$ consists of infinitely
      degenerate eigenvalues.
    \item If $\alpha = p/q$, where $p$ and $q$ are relatively prime, $\alpha j + \theta + \frac{1}{2} \notin \Z$
      for all $j = 0,\ldots, q-1$ and assumptions of part (a) do not hold, then $-\Delta_{\gamma, \mat A}$
      has infinitely degenerate eigenvalues at the points of $\{n^2 |\, n\in\N\}$ and
      an absolutely continuous part of the spectrum such that in each interval $(-\infty, 1)$
      and $\big(n^2, (n+1)^2\big)$, $n \in \N$ it consists of $q$ closed intervals possibly touching at the endpoints.
    \item If $\alpha = p/q$, where $p$ and $q$ are relatively prime, and $\alpha j + \theta + \frac{1}{2} \in \Z$ for some $j = 0,\ldots, q-1$,
      then the spectrum $-\Delta_{\gamma, \mat A}$ is of pure point type and such that in each interval $(-\infty, 1)$
      and $\big(n^2, (n+1)^2\big)$, $n \in \N$ there are exactly $q$ distinct eigenvalues and
      the remaining eigenvalues form the set $\{n^2 |\, n\in\N\}$. All the eigenvalues are infinitely degenerate.
    \item If $\alpha \notin \Q$, then $\sigma(-\Delta_{\gamma, \mat A})$ does not depend on $\theta$ and it
      is a disjoint union of the isolated-point family $\{n^2 |\, n\in\N\}$ and Cantor sets, one inside each interval
      $(-\infty, 1)$ and $\big(n^2, (n+1)^2\big)$, $n \in \N$. Moreover, the overall Lebesgue measure
      of $\sigma(-\Delta_{\gamma, \mat A})$ is zero.
  \end{enumerate}
\end{theorem}
\begin{proof}
  For parts (a), (b) and (c) one uses Theorem~\ref{th:spec_LA_for_rationals}, Proposition~\ref{prop:spec:basic_for_original_op}
  and properties of function $\eta$ discussed before Proposition~\ref{prop:spec:basic_for_original_op}.
  The conclusion is implied by the bicontinuity of $\eta$ on each set $I_n$, $n \in \N$,
  and by the fact that in (b), (c) $\sigma(L_{\alpha,\theta}) \subset (-4, 4)$ follows from \eref{eq:spec_LA_bound}.
  Under the assumptions of (a), $\sigma(L_{\alpha,\theta}) = [-4, 4]$,
  and thus $\eta^{(-1)}\big(L_{\alpha,\theta}\big) = [0,\infty)$, see also Figure \ref{fig:eta_fun}. The fact that the points $n^2$, $n \in \N$
  are contained in $\sigma_{ac}(-\Delta_{\gamma, \mat A})$ results from the closeness of the absolutely continuous spectrum.

  Finally, let us prove part (d). By Theorem~\ref{th:spec_LA_irrational_cantor}, $\sigma(L_{\alpha,\theta})$ is a Cantor set with Lebesgue measure zero.
  From \eref{eq:spec_LA_bound} it follows again that $\sigma(L_{\alpha,\theta}) \subset (-4, 4)$.
  Hence, since $\eta$ is bicontinuous in each set $I_n$, $n \geq 0$, the preimage $\sigma_n = f^{(-1)}\big(\sigma(L_{\alpha,\theta})\big) \cap I_n$
  (using the notation from Proposition~\ref{prop:spec:basic_for_original_op}) is again a Cantor set contained in $(-\infty, 1)$ for $n=0$
  and in $\big(n^2, (n+1)^2\big)$ for $n \in \N$, respectively. It is easy to see that the Lebesgue measure of $\sigma_n$ is zero for every $n \geq 0$
  which implies that it is zero for the whole set. Now the sought assertion follows from Proposition~\ref{prop:spec:basic_for_original_op}.
\end{proof}
\begin{remark}
  It follows from the previous theorem that the eigenvalues $\{n^2 |\, n\in\N\}$
  are isolated points of the spectrum of $-\Delta_{\gamma, \mat A}$ if and only if
  $\gamma \neq 0$ or $\alpha \notin \Z$ or $\theta \notin \Z$.
\end{remark}
Finally, we may apply the very recent result of Last and Shamis
\cite{Last2016} which says that there is a dense $G_\delta$ set of
$\alpha$'s, for which the Hausdorff dimension of the spectrum of
$H_{\alpha, 2, \theta}$ equals zero, $\dim_H \sigma(H_{\alpha, 2,
\theta}) = 0$, see e.g. \cite{falconer, mattila} for the definitions
of Hausdorff measure and dimension. This result may be applied to
the spectrum of $-\Delta_{\gamma, \mat A}$ as a consequence of the
following proposition.
\begin{proposition}
  Let $A_j = \alpha j + \theta$ for some $\theta \in \R$, $\alpha \notin \Q$, and every $j\in\Z$.
  Then $\dim_H\sigma(-\Delta_{\gamma, \mat A}) = \dim_H\sigma(L_{\alpha,\theta})$.
\end{proposition}
\begin{proof}
  It follows from \eref{eq:spec_LA_bound} that $\sigma(L_{\alpha,\theta}) \subset (-4, 4)$.
  By the discussion preceding Proposition \ref{prop:spec:basic_for_original_op} and with the same notation, it follows that for any $n \geq 0$,
  $\sigma_n$ is contained in some closed subinterval $J_n$ of $I_n$.
  Moreover, for $n>0$ the function $\eta$ is bi-Lipschitz on $J_n$.
  Thus the inverse $(\eta|_{J_n})^{(-1)}$ of its restriction on $J_n$ is again bi-Lipschitz.
  Hence $\sigma_n$ is the image of $\sigma(L_{\alpha,\theta})$ under bi-Lipschitz function $(\eta|_{J_n})^{(-1)}$.
  It is a known fact, that bi-Lipschitz mappings
  preserve Hausdorff dimension, see e.g. \cite[Corollary 2.4]{falconer}. Hence $\dim_H(\sigma_n) = \dim_H \sigma(L_{\alpha,\theta})$ for all $n > 0$.
  For $n = 0$ we may argue similarly for any closed set contained in $J_0 \setminus \{0\}$. The point $0$ should be omitted
  since $\eta$ is not bi-Lipschitz on open sets containing zero. Let $H_0$ be a neighbourhood of $0$.
  Then $\sigma_0 \setminus H_0$ is an image of $\sigma(L_{\alpha,\theta}) \setminus \eta(H_0)$ under a bi-Lipschitz function  $(\eta|_{J_0})^{(-1)}$.
  Since $H_0$ was arbitrary, it follows that $\dim_H(\sigma_0) = \dim_H \sigma(L_{\alpha,\theta})$.
  Finally, since countable sets have Hausdorff dimension zero, the countable stability, see e.g. Section 2.2 in \cite{falconer}, of Hausdorff measures yields the assertion.
\end{proof}
Thus, by \cite[Theorem 1]{Last2016} and
\eref{eq:spec_set_equality_LA_mathieu}, one more  assertion
follows.
\begin{corollary}
  Let $A_j = \alpha j + \theta$ for some $\alpha, \theta \in \R$ and every $j\in\Z$.
  There exist a dense $G_\delta$ set $S$, such that for every $\alpha \in S$,
  \[
    \dim_H \sigma(-\Delta_{\gamma, \mat A}) = 0
  \]
  for all $\theta$.
\end{corollary}

%
%
\section{Concluding remarks}\label{sec:concl}

To conclude, recall first that for any irrational $\alpha$ and (Lebesgue) almost all $\theta$ the spectrum of the almost Mathieu operator $H_{\alpha, 2, \theta}$ is purely singularly continuous. This is a part of the more general Aubry-Andr\'{e} conjecture proven by Jitomirskaya \cite{Jitomirskaya99}. This fact motivates us to the question whether for any irrational $\alpha$ the spectrum of $L_{\alpha,\theta}$ has the same property, i.e. whether it is purely singularly continuous for Lebesgue a.e. $\theta$.

A deeper question concerns the physical meaning of the model that involves a magnetic field changing linearly along the chain. A philosophical answer could be, according the known quip of Bratelli and Robinson, that ``validity of such idealizations is the heart and soul of theoretical physics and has the same fundamental significance as the reproducibility of experimental data''. On a more mundane level, one can note that the spectral behaviour will not change if the linear field is replaced by a quasiperiodic one which changes in a saw-tooth-like fashion as long as the jumps coincide with the graph vertices. This also opens an interesting question about the spectral form and type in case when the saw-tooth shape is replaced by another periodic or quasiperiodic function.

\section*{Acknowledgements}
We are obliged to Konstantin Pankrashkin for his constructive
criticism which helped to improve the manuscript. The research was
supported by the Czech Science Foundation (GA\v{C}R) within the
project 17-01706S.


\end{document}